\documentclass[journal,10pt]{IEEEtran}
\IEEEoverridecommandlockouts

\usepackage{graphicx}
\DeclareGraphicsExtensions{.pdf,.png,.jpg, .tif}
\usepackage{amsmath}
\usepackage{latexsym}
\usepackage{graphicx,color}
\usepackage{bm}
\usepackage{amssymb}
\usepackage{array}
\usepackage{setspace}
\usepackage{fancyhdr}
\usepackage{float}
\usepackage{algorithm}
\usepackage{algorithmicx}
\usepackage{algpseudocode} 
\newcommand{\subparagraph}{}
\usepackage{titlesec}
\usepackage[stable]{footmisc}
\usepackage[table,xcdraw]{xcolor}
\usepackage{parskip}
\usepackage{url}
\usepackage{hyperref}
\usepackage{cite}
\usepackage{balance}
\newtheorem{theorem}{Theorem}

\usepackage{csquotes}
\MakeOuterQuote{"}
\makeatletter
\newcommand{\vast}{\bBigg@{2.6}}
\newcommand{\Vast}{\bBigg@{4.6}}
\newcommand{\ignore}[1]{}
\makeatother

\begin{document}

\title{Asynchronous Task Allocation for Federated and Parallelized Mobile Edge Learning}

\author{Umair Mohammad,~\IEEEmembership{Student Member,~IEEE}, Sameh Sorour,~\IEEEmembership{Senior Member,~IEEE}
	
	\thanks{Umair Mohamed is with the Department of Electrical and Computer Engineering, University of Idaho, Moscow, ID, 83843, USA e-mail: (moha2139@vandals.uidaho.edu}
	\thanks{Sameh Sorour is with the School of Computing, Queen's University, Kingston, ON, 83843, Canada e-mail: (sameh.sorour@queensu.ca)}
	\thanks{Manuscript received XXX, XX, 2020; revised XXX, XX, 2020.}}


\maketitle

\begin{abstract}
This paper proposes a scheme to efficiently execute distributed learning tasks in an asynchronous manner while minimizing the gradient staleness on wireless edge nodes with heterogeneous computing and communication capacities. The approach considered in this paper ensures that all devices work for a certain duration that covers the time for data/model distribution, learning iterations, model collection and global aggregation. The resulting problem is an integer non-convex program with quadratic equality constraints as well as linear equality and inequality constraints. Because the problem is NP-hard, we relax the integer constraints in order to solve it efficiently with available solvers. Analytical bounds are derived using the KKT conditions and Lagrangian analysis in conjunction with the suggest-and-improve approach. Results show that our approach reduces the gradient staleness and can offer better accuracy than the synchronous scheme and the asynchronous scheme with equal task allocation.
\end{abstract}

\ignore{}

\IEEEpeerreviewmaketitle

\section{Introduction}

Mobile edge computing (MEC) is rapidly re-defining infrastructure with the world moving towards smart cities, smart grids and the internet of everything (IoE). It is expected that by 2022, 41 billion IoE devices will be connected to the internet and will generate up to 800 zettabytes of data \cite{N202003_vXchnge_IoT_data}. The expectation is that the time-critical nature of such data would force us to do 90\% of analytics on the edge servers and the nodes themselves (mobile phones, traffic cameras, UAV's and autonomous vehicles) \cite{Ref2_Original}. 

For example, a wireless edge system may comprise a road-side unit (RSU) connected via dedicated short range communication (DSRC) to a set of on-board units (OBUs) on cars jointly computing a task. 
This paradigm of edge processing has been supported by the latest works in literature about MEC and Hierarchical-MEC (H-MEC) \cite{P0_Letaief_Survey, RefsPap_MEC126, Moha1812:Multi}. One example of such processing is machine learning (ML), which is used in all types of applications such as object recognition and image segmentation; applications will form the basis for edge AI. 

Performing ML in a distributed manner, a.k.a Distributed Learning (DL) is attracting a lot of attention in the ML community in general. In particular, the deployment of DL models over devices connected via wireless edge networks, which can also be called Mobile Edge Learning (MEL) is of increasing interest to researchers \cite{Tuor_01_AdaptiveControl, Wang2019,2019arXiv190907972C,2019arXiv191102417Y,Mohammad2019a,Moha2020b}. Typically, in such schemes, the orchestrator waits for all learners to complete an equal number of iterations of the ML training algorithm and hence, we call this the synchronous approach. The idea behind this approach is to maximize accuracy by minimizing the discrepancy or 'staleness' among the gradients of each learner. 
Recently, some work has been carried out on allowing some staleness so that powerful devices with good communication links may actually provide a faster validation accuracy progression \cite{StalenessAwarePaper, AsyncFedOpt_proof}.

The works of \cite{Tuor_01_AdaptiveControl, Wang2019} aimed to optimize the number of local epochs per node with respect to total global iterations in generic resource-constrained edge environments. However, these works do not take into account the heterogeneous nature of communication and computation in MEC's. Recently, the works of \cite{2019arXiv190907972C,2019arXiv191102417Y} have optimized resource allocation while maintaining accuracy. However, they do not investigate the impact of batch allocation. In contrast, the works of \cite{Mohammad2019a,Moha2020b} and  investigates the impact of maximizing the number of local updates on the learning accuracy by optimizing the size or portion of the local dataset used. Although results show significant gains in achieving a higher validation accuracy,. However, there may still be room for improvement as certain devices may be idle for long times and can do a higher number of updates which may raise the overall accuracy.   

To the best of the authors' knowledge, this work is the first attempt to have a staleness aware algorithm for asynchronous MEL. Here, we emphasize that our model is different than the models in \cite{StalenessAwarePaper, AsyncFedOpt_proof} such that the system is asynchronous in terms of the number of updates each learner is allowed within one global cycle which will be constrained by time. This will make sure that the aggregation is done uniformly for all learners without being affected by stragglers. The novelty is also in the fact task allocation and number of local updates per learner will be jointly optimized in order to minimize the staleness among gradients in order to achieve a higher validation accuracy and hence, it is heterogeneity aware (HA). 

The formulated optimization problem is shown to be an integer quadratically-constrained linear program (IQCLP) which is relaxed to a non-convex QCLP. Analytical approximate solutions are derived based on the KKT conditions and Lagrangian analysis followed by a suggest-and-improve (SAI) approach. The merits of the proposed solution will be compared against the heterogeneity unaware (HU) approach in \cite{Wang2019} and the synchronous method in \cite{Mohammad2019a}. 


\section{System Model for Asynchronous MEL}
\label{Section02__SystemModelParameters}

There are two approaches possible for MEL: parallelized learning (PL) and federated learning (FL). In the first case, the global orchestrator offloads randomly picked subsets to each learner whereas in the second scenario, the learners operate on locally stored datasets. PL can be utilized when the orchestrator does not have enough resources to learn on the complete data and thus, distributes the learning tasks to a set of learners. On the other hand, in FL, the learners may collect their own data and take advantage of learning on a larger dataset while maintaining privacy.

Consider a set of $K$ learners in which learner $k$, $k \in \mathcal{\kappa} \text{ where } \mathcal{\kappa} = \{ 1, 2, \dots, K\}$ trains its local learning model, or learns from a batch of size $d_k$ data samples by performing $\tau_k$ learning epochs/updates/iterations. The total size of all batches is denoted by $d=\sum_{k=1}^K d_k$. Fig. \ref{figure0label} illustrates the described MEL system. The objective is to minimize the local loss function in order to minimize global loss such that accuracy is maximized \cite{Tuor_01_AdaptiveControl} .  
	\begin{figure}[t]
	\centering
	\includegraphics[scale=0.4]{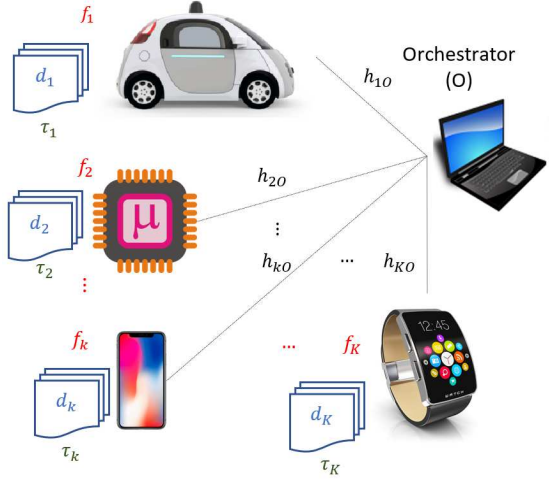}
	\caption{Asynchronous MEL Model}
	\label{figure0label}
\end{figure}

In an asynchronous environment, each learner will perform $\tau_k$ epochs and forward its updated set of parameters $\mathbf{w}_k$ to the orchestrator. The orchestrator will aggregate the model parameters to form a globally optimized set $\mathbf{w}$ and send back the updated model to each learner in the next cycle. Based on the channel conditions and the compute capability of each individual device, it will also offload $d_k$ samples (PL) or assign a value for the subset size $d_k$ (FL) to each node $k$. In both scenarios, it will also assign the number of updates $\tau_k$ to perform at each node. The learners will apply the ML algorithm to their assigned dataset and the process continues.

The time taken for offloading the optimal model and the partial dataset to each node, then for each learner to perform the ML task and send back the locally updated model, and for the orchestrator to perform global aggregation is defined as $t_k$. This time is bounded by $T$ and usually excludes the global aggregation process because it requires a lot less time compared to transmission and ML execution; $T$ is known as the global cycle clock. These global updates are performed a certain number of times. In contrast to the synchronous case, the asynchronous model allows each learner to perform different number of $\tau_k$ updates. Given the above description, the times of each learner $k$, $\forall~k$, whose sum must be bounded by the global update clock $T$, can be detailed as follows:

\begin{equation}
t_{k}^S \ignore{&= \dfrac{B_k^{data}+B_k^{model}}{R_{k}} \nonumber\\
	&} = \dfrac{d_k\mathcal{F}\mathcal{P}_d + \mathcal{P}_m \left(d_k\mathcal{S}_d+\mathcal{S}_m\right)}{W\log_2\left(1+\frac{P_{ko} h_{ko}}{N_0}\right)}
\label{eq:edge_time_sending}
\end{equation}
\begin{equation}
t_k^C = \dfrac{\tau_k d_k C_m}{f_k}
\label{Eq_9_localTimeExecution}
\end{equation}
\begin{equation}
t_{k}^R = \ignore{\dfrac{B_k^{model}}{R_{k}} = } \dfrac{\mathcal{P}_m \left(d_k\mathcal{S}_d+\mathcal{S}_m\right)}{W\log_2\left(1+\frac{P_{ko} h_{ko}}{N_0}\right)}
\label{eq:edge_time_receiving}
\end{equation}

Please note that equations (\ref{eq:edge_time_sending})-(\ref{eq:edge_time_receiving}) represent the following: $t_k^S$ denotes the time taken to transmit the global parameter set $\mathbf{w}$ and the allocated batch to learner $k$ \footnote{Note that the the first term of the numerator will not exist for FL.}, the time needed by learner $k$ to execute one update of the ML algorithm, and $t_k^R$ which is the time needed for learner $k$ to send its updated local parameter matrix $\tilde{\mathbf{w}}_k$ back to the orchestrator, respectively.

The first and second terms in the numerator of (\ref{eq:edge_time_sending}) give the total sizes in bits of the transmitted data and the optimal model parameter set $\mathbf{w}$, respectively. The total transmitted data size in bits per sample is a product of the number of features $\mathcal{F}$ and the storage precision/method $\mathcal{P}_d$. In the second term of the numerator, $\mathcal{P}_m$ represents the precision with which model parameters are stored, $\mathcal{S}_d$ and $\mathcal{S}_m$ each describe the relationship of the ML model size in bits to the allocated batch size and the ML model architecture, respectively. The denominator represents the achievable rate with respect to the channel parameters where $W$ is the available bandwidth, $N_0$ is the noise power spectral density, $P_{ko}$ is the available transmission power and $h_{ko}$ represents the channel parameters between the learner and the orchestrator. It is assumed that the channel is perfectly reciprocal within one global cycle. 

The time needed by learner $k$ to execute one update of the ML algorithm is given by (\ref{Eq_9_localTimeExecution}) where $\mathcal{C}_m$ is the complexity of the learning technique in terms of clock cycles required and $f_k$ is the processing power of each learner $k$ in clocks per second. ML algorithms typically go over all features sequentially for each data sample at a time (or epoch), so, the time for one update for one sample is multiplied by $\tau_k$ and $d_k$. (In case of batch learning at the local node, the complexity expression changes but $t_k^C$ remains the same). Thus, the total time $t_k$ taken by learner $k$ to complete the above three processes is equal to:
\begin{equation}
\label{Eq_12_timeForLearnerk}
t_k = t_{k}^S+\tau_k t_k^C + t_k^R
\end{equation}
The total time $t_k$ can be re-written as a quadratic expression of the optimization variables $\tau$ and $d_k$ as shown in (\ref{Eq_16_C1isCompact})\footnote{Note that for FL, the first term of the numerator in $C_k^1$) will not exist.}.

\begin{align}
t_k &= \dfrac{d_k\mathcal{F}\mathcal{P}_d + 2\mathcal{P}_m \left(d_k\mathcal{S}_d+\mathcal{S}_m\right)}{W\log_2\left(1+\frac{P_{ko} h_{ko}}{N_0}\right)} + \tau_k \dfrac{d_k\mathcal{C}_m}{f_k}   \nonumber\\
&= C_k^2  \tau_k d_k + C_k^1 d_k + C_k^0
\label{Eq_16_C1isCompact}
\end{align}
The quadratic, linear and constant coefficients are given by $C_k^2$, $C_k^1$ and $C_k^0$, respectively, where, $C_k^2 = \frac{\mathcal{C}_m}{f_k}$, $C_k^1 = \frac{\mathcal{F}\mathcal{P}_d+2\mathcal{P}_m\mathcal{S}_d}{W\log_2\left(1+\frac{P_{ko} h_{ko}}{N_0}\right)}$, and $C_k^0 = \frac{2\mathcal{P}_m\mathcal{S}_m}{W\log_2\left(1+\frac{P_{ko} h_{ko}}{N_0}\right)}$.

\section{Problem Formulation}
The staleness $s$ between any two learners can be described as the difference between the number of local ML iterations each has performed as shown below:
\begin{equation}
s = \lvert \tau_k - \tau_l \rvert
\label{eq:staleness_def}, ~k\in\kappa~ \text{\&} ~l \in \left\{\mathcal{\kappa} \mid l < k ~\forall~k \right\}
\end{equation}
It has been shown in the literature that the loss function of SGD-based ML is minimized (and thus the learning accuracy is maximized) by minimizing the staleness between the gradients in Asynchronous SGD \cite{StalenessAwarePaper,AsyncFedOpt_proof}. Although our model is different, we show in Appendix A\ignore{ of \cite{Mohammad2019_AsynArxiv}} a lower staleness in our model can reduce model divergence and improve accuracy.
 
Overall, the maximum staleness has to be minimized while satisfying the global cycle time constraint.
Clearly, the relationship between $t_k$ and the optimization variables $d_k$ and $\tau$ is quadratic. Furthermore, the optimization variables $\tau$ and $d_k$ $\forall~k$ are all non-negative integers. Consequently, the problem can be formulated as an ILP with quadratic and linear constraints as follows: \footnote{Note that the problem type and solution remain the same with different $C^1_k$ expressions for the two distinct scenarios of FL and OL.}
\begin{subequations}
	\begin{align}
	&\operatornamewithlimits{min}_{\tau_k,~{d}_k~\forall~k}  \quad \max \{s\} \\
	& \quad \nonumber\\
	\text{s.t. }\qquad & C_k^2  \tau_k d_k + C_k^1 d_k + C_k^0 = T, \quad k = 1,\ldots,K \label{orignial-time-const}\\
	& \sum_{k = 1}^{K}d_k = d \label{orignial-batch-const}\\ 
	& \tau_k \in \mathbb{Z}_+, \quad k \in \mathcal{\kappa} \label{orignial-tau-const}\\
	& d_k \in \mathbb{Z}_+, \quad k \in \mathcal{\kappa} \label{orignial-d-const} \\
	& d_l \leq d_k \leq d_u, \quad k \in \mathcal{\kappa} \label{bounds-d-const}
	\end{align}
	\label{Eq_13_OurProb}
\end{subequations}
Constraint (\ref{orignial-time-const}) guarantees that $t_k = T$ $\forall~k$, which means that all devices work for the full allotted time though they may perform different number of local updates. Constraint (\ref{orignial-batch-const}) ensures that the sum of batch sizes assigned to all learners is equal to the total dataset size that the orchestrator needs to analyze. Constraints (\ref{orignial-tau-const}) and (\ref{orignial-d-const}) are simply non-negativity and integer constraints for the optimization variables. Please note that the solutions of (\ref{Eq_13_OurProb}) having any $\tau_k$ and/or $d_k$ being zero represent conditions where MEL is not feasible for learner $k$. Constraint (\ref{bounds-d-const}) bounds the number of data points dispersed to each learner in order to ensure that each node performs learning on some part of a dataset and no single node is burdened with too many data samples. Therefore, the problem is an ILPQC, which is well-known to be NP-hard \cite{QIP_NP_ArXiV}. We will thus propose a simpler solution to it through the relaxation of the integer constraints in the next section.

\section{Proposed Solution}\label{Section3_Solution}

\subsection{Problem Transformation and Relaxation}
\begin{subequations}
	\begin{align}
	&\operatornamewithlimits{min}_{\tau_k,~{d}_k~\forall~k} z \\
	& \quad \nonumber\\
	\text{s.t. } \qquad & \lvert \tau_k-\tau_l \rvert \leq z, ~ k \in \kappa ~ \& ~l \in \left\{\kappa \mid l > k ~ \forall~k \right\} \label{relaxed-staleness-const} \\
	& C_k^2  \tau_k d_k + C_k^1 d_k + C_k^0 = T, \quad k \in \mathcal{\kappa} \label{relaxed-time-const}\\
	& \sum_{k = 1}^{K}d_k = d \label{relaxed-batch-const}   \\ 
	& \tau_k \geq 0, \quad k \in \kappa \ \label{relaxed-tau-const}\\
	& d_l \leq d_k \leq d_u, \quad k \in \mathcal{\kappa} \label{relaxed-bounds-d-const}
	\end{align}
	\label{Eq_13_RelaxedProblem}
\end{subequations}

As described in the previous section, the problem of interest is NP-hard due to its integer decision variables. We simplify the problem by relaxing the integer constraints in (\ref{orignial-tau-const}) and (\ref{orignial-d-const}), solving the relaxed problem, then flooring the obtained real results back into integers. The problem is re-formulated by applying a min-max transformation and relaxing of the integer constraints as shown in (\ref{Eq_13_RelaxedProblem}).

We introduce a slack variable $z$ and add an additional constraint to ensure the staleness is less than $z$ which will guarantee that the maximum staleness is minimized. Please note that constraint (\ref{orignial-d-const}) has been eliminated due to the lower bound on $d_k$. The above resulting program becomes a linear program with quadratic constraints. This problem can be solved by using interior-point or ADMM methods using commercial solvers. From the analytical viewpoint, the associated matrices to each of the quadratic constraints in (\ref{relaxed-time-const}) can be written in a symmetric form. However, these matrices will have two non-zero values that are positive and equal\ignore{ to each other}. The eigenvalues will thus sum to zero, which means these matrices are not positive semi-definite, and hence, the relaxed problem is non-convex. Consequently, we cannot derive the optimal solution of this problem analytically. Hence, we will calculate upper bounds using Lagrangian analysis followed by an improve step to reach the feasible solution. 

\subsection{Upper Bounds using Lagrangian Analysis and KKT conditions}
Let $\mathbf{\tau} = \{\tau_1, \ldots, \tau_k, \ldots, \tau_K\}$ and $\mathbf{d} = \{d_1, \ldots, d_k, \ldots, d_K\}$. The Lagrangian of the relaxed problem is given by (\ref{lagrangian}). The Lagrangian multipliers associated with the time constraints of the $K$ learners in (\ref{relaxed-time-const}), the total batch size constraint in (\ref{relaxed-batch-const}), the non-negative constraints of the number of epochs at each node $\tau_k$ in (\ref{relaxed-tau-const}) and the lower and upper bounds in (\ref{relaxed-bounds-d-const}) are given byL: $\lambda_k$ $k\in\kappa$, $\omega$, and $\nu_k$/$\nu_k^\prime ~ \forall ~ k\in \kappa$, respectively. The multipliers $\mu_n$ and $\mu_n^\prime$ $n\in\{1,\dots,N\}$ are associated with (\ref{relaxed-staleness-const})\ignore{the staleness between each two learners being less than the slack variable which we minimize over $\tau_k$ and $d_k$}. Note that the absolute value constraint in (\ref{relaxed-staleness-const}) can be decoupled as $\tau_k-\tau_l\leq z$ and $\tau_l-\tau_k\leq z$, $k \in \mathcal{\kappa} ~ \& ~l \in \left\{\mathcal{\kappa} \mid l > k ~ \forall~k \right\}$.
\begin{multline}
\label{lagrangian}
L\left(z, \mathbf{\tau}, \mathbf{d}, \mathbf{\lambda}, \mathbf{\alpha}, \mathbf{\omega}, \mathbf{\nu}, \mathbf{\nu^{\prime}}, \mathbf{\mu}, \mathbf{\mu^{\prime}} \right) = z + \\ \sum_{k = 1}^{K} \lambda_k\left(C_k^2  \tau_k d_k + C_k^1 d_k + C_k^0 - T\right) + \alpha_k\tau_k + \\ \omega\left(\sum_{k = 1}^{K}d_k - d\right) + \sum_{k=1}^K \nu_k \left(-d_k+d_l\right) + \sum_{k=1}^K \nu_k^{\prime} \left(d_k-d_u\right) + \\ \sum_{n=1}^N \mu_n \left(-z+\tau_{c_{n,1}}-\tau_{c_{n,2}}\right) + \\ \sum_{n=1}^N \mu_n^{\prime} \left(-z-\tau_{c_{n,1}}+\tau_{c_{n,2}}\right)    
\end{multline}

The matrix $\mathbf{c} \in \mathbf{\mathbb{R}}^{N \times 2}$ where N is the number of possibilities of mutual staleness for K set of users, i.e. $N = \binom{K}{2}$. For example, for a set of 4 users, $N=6$ and the matrix of possibilities will be:
\begin{equation}
\mathbf{c}=
\left[ {\begin{array}{cccccc}
	1 & 1 & 1 & 2 & 2 & 3\\
	2 & 3 & 4 & 3 & 4 & 4\\
	\end{array} } \right]^T
\label{eq_cexample}
\end{equation}

Using the KKT conditions $\nabla L_\mathbf{x} = 0$, the following theorem gives a way to find the optimal values of $\tau_k$ and $d_k$ using the Lagrange multipliers.
\begin{theorem}
	The optimal number of updates each user node can perform $\tau_k$ can be given by:
	\begin{equation}
	\tau_k^* = -\dfrac{\lambda_k C_k^1 + \nu_k+\nu_k^\prime + \omega}{\lambda_k C_k^2} \qquad \forall~k
	\label{Eq_36_AnalBound}    
	\end{equation}
	Moreover, the optimal value of $ d_k$ can be given by the following equation:
	\begin{equation}
	d_k^* = -\dfrac{u_k+u_k^\prime+\alpha_k}{\lambda_k C_k^2} \qquad \forall~k
	\label{Eq_30_BoundOnDk}
	\end{equation}
	Each element of the vectors $\mathbf{u}$ and $\mathbf{u}^\prime$ is a function of the Lagrange multipliers $\mu_n$ and $\mu_n^\prime$. Please refer to the proof. 
	\label{theorem1}
\end{theorem}
\begin{proof}
	The proof of this theorem can be found in Appendix B. The details about how to obtain $\mu$ and $\mu^\prime$ can be found in Appendix C.
\end{proof} 

As suspected, due to the relaxed problem being non-convex with quadratic constraints, in some situations, the approach described above resulted in infeasible solutions. In that case, we performed constraint checks and then used the initial solution to carry out \ignore{ for $d_k^*$ $\forall~k$ in (\ref{Eq_30_BoundOnDk}) and $\tau^*$ from solving (\ref{Eq_36_AnalBound}) should undergo} suggest-and-improve (SAI) steps to reach a feasible solution. The set of feasible solutions was used as a starting point to the less complex improve method in order to reach the optimal solution.  

\section{Results}
This section presents the results of the proposed scheme by testing in MEL scenarios emulating realistic edge node environments and learning. We show the merits of the proposed HA solution compared to performing asynchronous learning with the HU method in terms of staleness and learning. For the staleness, one of the the metrics will be maximum staleness as described in (\ref{eq:staleness_def}). In addition, we would like to introduce average staleness as shown in (\ref{eq:staleness_def_avg}) which will give a measure of the mutual staleness between every two learners for all learners. The metric for evaluating the learning performance is validation accuracy.
\begin{equation}
s_{avg} = \dfrac{1}{N}\sum_{n=1}^{N}\lvert \tau_{c_{n,1}}-\tau_{c_{n,2}} \rvert 
\label{eq:staleness_def_avg}
\end{equation}
\ignore{as well as the OPTI-based numerical solution to the relaxed problem in (\ref{Eq_13_RelaxedProblem}).}
\ignore{ We will first introduce the simulation setting, then present the testing results.}

\subsection{Simulation Environment, Dataset, and Learning Model}
The simulation environment considered is an indoor environment which emulates 802.11-type links between the edge nodes that are located within a radius of 50m. We assume that that approximately half of the nodes have the processing capabilities of typical computing devices such as desktops/laptops and the other half consists of industrial micro-controller type nodes such as a Raspberry Pi. The employed channel model is summarized in Table \ref{Table_1_OfParameters}. 
\begin{table}[!t]
	\centering
	\small
	\caption{List of simulation parameters}
	\begin{tabular}{|l|l|}
		\hline
		Parameter                                 & Value                         \\ \hline
		Attenuation Model                           & $7+2.1\log(R)$ dB \cite{Moha1812:Multi}         \\ \hline
		System Bandwidth $B$                        & 100 MHz                              \\ \hline
		Node Bandwidth $W$                          & 5 MHz                               \\ \hline
		Device proximity $R$                        & 50m                                 \\ \hline
		Transmission Power $P_k$                   & 23 dBm                              \\ \hline
		Noise Power Density $N_0$                  & -174 dBm/Hz                         \\ \hline
		Computation Capability $f_k$             & 2.4 GHz and 700 MHz             \\ \hline
		MNIST Dataset size $d$                 & 60,000 images                        \\ \hline
		MNIST Dataset Features $\mathcal{F}$  & 784 ($~28 \times 28~$) pixels     \\ \hline
	\end{tabular}
	\label{Table_1_OfParameters}
\end{table}

As a benchmark, the MNIST dataset \cite{MNIST_IEEE} is used to evaluate the proposed scheme. The training data comprises 60,000 28x28 pixel images contributing 784 features each. The ML algorithm tested is the a simple deep neural network with the following configuration $[784, 300, 124, 60, 10]$. The size of the resulting model is 8,974,080 bits, which is fixed for all edge nodes, and the forward and backward passes will require 1,123,736 floating point operations \cite{Mohammad2019a}.

\subsection{Staleness Analysis}
\begin{figure}[t]
	\centering
	\includegraphics[scale=0.56]{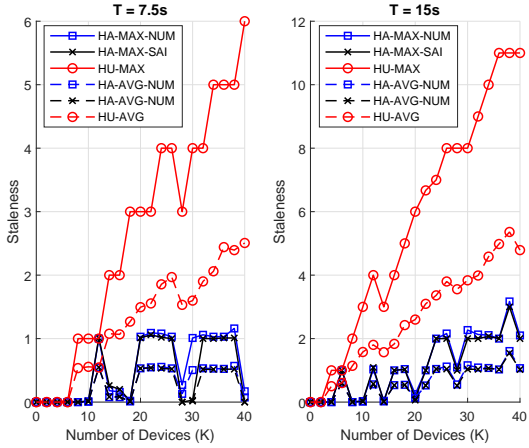}
	\caption{Maximum and Average Staleness vs $K$ for $T=7.5$s and $T = 15$s.}
	\label{figure2label}
\end{figure}
Fig \ref{figure2label} shows the maximum and average staleness versus the number of nodes for global cycle times of 7.5s and 15s for the HA asynchronous scheme from the numerical (HA-MAX/AVG-NUM) and SAI-based solutions (HA-MAX/AVG-NUM) and the HU scheme (HA-MA/AVG) as well. In general, the SAI-based approach gives similar staleness to the numerical solution. The general trend is that as the number of updates $\tau_k$ increase, the staleness tends to increase. However, for $T=7.5s$, the maximum staleness does not exceed around 1 and the average staleness is between 0.4-0.6 as K increases for the proposed HA scheme. For example, for our HA scheme with 20 users at $T = 7.5$s, the maximum staleness is 1 compared to 4 for the HU which is 400\% higher and the average staleness is 1.5 compared to 0.5 for our scheme which is 300\% higher. One curious aspect to note is that for certain specific number of learners or $K$, the asynchronous scheme is able to find an optimal solution where the staleness is zero. One such example is $K=$ 14 for $T=15$s and $K=18$ for $T=7.5$s.  

\subsection{Validation Accuracy}
Fig. \ref{figure4label} shows the learning accuracy for a system with a limit on the global cycle time of $T=15$s consisting of 10, 15 and 20 learners, respectively. For example, in the case with 10 learners, the proposed HA scheme (HA-Asyn) achieves an accuracy of 95\% within 4 updates or 1 minute of learning as compared to the HA synchronous scheme (HA-Sync) in \cite{Mohammad2019a} which requires 8 updates, in other words, we obtain a gain of 50\%. In contrast, the HU scheme (HU-Asyn) fails to converge or even achieve a 95\% accuracy. An accuracy of 95\% is achieved by our scheme within 3 updates with 15 users whereas the other schemes require 4 updates; which gives us a gain of 25\%. Moreover, our scheme achieves an accuracy of 97\% within 8 updates whereas the other two methods require 10 global cycles leading to a gain of 25\%. 

A similar gain is achieved for a system with 20 learners for the 95\% accuracy mark. For the case of 97\% accuracy, our scheme requires 7 updates whereas the ETA needs 11 cycles, representing a gain of of about 64\%. On the other hand, the synchronous scheme requires 8 updates which translates to a gain of only 12.5\%. The gain appears marginal compared to the synchronous scheme because as the number of users increase, each learner has to process less data which means a larger number of synchronized updates can be done even in heterogeneous conditions. In contrast, the gain is significant compared to the ETA scheme because the staleness for ETA increases significantly versus $K$ for a fixed global cycle $T$.
\begin{figure}[t]
	\centering
	\includegraphics[scale=0.56]{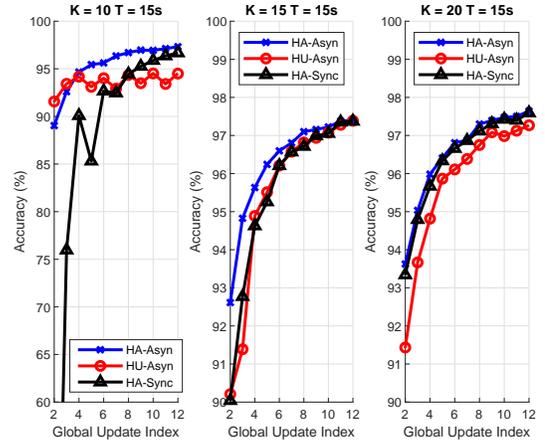}
	\caption{Learning accuracy progression after global update cycles for $K=10, 15 \text{ and } 20$ for $T = 15$s}
	\label{figure4label}
\end{figure}
\ignore{
	Fig. \ref{figure5label} shows the accuracy with 20 users for a global cycle time of $T=7.5$s and a total learning time of 1 minute. The results here are interesting in the sense that our scheme achieves an accuracy of 95\% in 5 cycles whereas the synchronous requires 7 cycles; in other words, a gain of 28.57\%. For the asynchronous case with equal task allocation, the validation accuracy jumps to above 92\% in the first 4 updates, but then saturates to around 94\%. After some deeper investigation, it was found out that the loss function saturates to a very high value as compared to both, the asynchronous and synchronous schemes with optimized task allocation.
	\begin{figure}[t]
		\centering
		\includegraphics[scale=0.58]{MNIST_Validation_K20_T7p5}
		\caption{Achievable number of local update cycles by all schemes (a) vs $K$ for $T=30$s and $60$s (b) vs $T$ for $K=10$ and $20$.}
		\label{figure5label}
	\end{figure}
}

\section{Conclusion}
This paper proposed a scheme to optimize batch allocation for asynchronous MEL by reducing the staleness among the gradients in the MEL system by minimizing the maximum difference between the number of updates done by each learners. The resulting optimization problem was an NP-hard IQCLP which was relaxed to a non-convex problem which was solved numerically and theoretically using Lagrangian analysis followed by the SAI approach. Through extensive simulations on the well-known MNIST dataset, the proposed scheme was shown to perform better than asynchronous ETA and the synchronous schemes in terms of learning accuracy.

\appendices
\section{Staleness and Model Divergence}
\ignore{The convergence bounds have been derived and well-discussed in \cite{Tuor_01_AdaptiveControl}. For completeness, we will present some of the important results here in order to support our analysis. }
Let us assume that a total of $L$ updates occur where a global aggregation occurs at any iteration $l$ for $l = 1,~\ldots,~L$. Between any global aggregation $g$ and $g+1$, in the synchronous version in \cite{Mohammad2019a}, each learner would have performed $\tau$ updates whereas they will each perform $\tau_k$ local updates in the proposed model. 
Let us assume, to facilitate the analysis, that the global aggregations occur at integer multiples of $\tau_m$ which is the maximum possible local learning iterations by the best performing learner. 
\ignore{I.e. the local updates occur at every iteration $l=1,\ldots,L$ and a global update will occur whenever $l = g\tau_m$ for $g=1,\ldots,G$.}
For any interval $[g]$ defined over $[g(\tau_m-1),g\tau_m]$, define an auxiliary global model denoted by $\hat{\mathbf{w}}$ which would have been calculated if a global update occurred as follows:
\begin{equation}
\hat{\mathbf{w}}_{[g]}[l] = \hat{\mathbf{w}}{[g]}[l-1]-\eta \nabla F(\hat{\mathbf{w}}_{[g]}[l-1]) 
\end{equation} 
The learning rate is given by $\eta$ and the loss function can be described by $F$.

Let the local model parameter set of learner $k$ be denoted by $\mathbf{w}_k$ and the local loss by $F_k(\mathbf{w}_k)$. Then, the optimal model at iteration $l$ can be obtained by:
\begin{equation}
\mathbf{w}[l] = \dfrac{1}{d}\sum_{k = 1}^K d_k \mathbf{w}_k[l]
\end{equation} 
The optimal $\mathbf{w[l]}$ will only be visible at an iteration $l$ such that a global aggregation occurs in that iteration. Then, the global loss can be defined by:
\begin{equation}
F(\mathbf{w}) = \dfrac{1}{d}\sum_{k = 1}^K d_k F_k(\mathbf{w})
\end{equation} 

The following assumptions are made about the local loss function $F_k(\mathbf{w})$ at learner $k$: $F_k(\mathbf{w})$ is convex, it is $\rho$-Lipschitz $\lVert F_k(\mathbf{w})-F_k(\bar{\mathbf{w}}) \rvert \leq \rho \lvert \mathbf{w}-\bar{\mathbf{w}} \rvert$, and $\beta$-smooth $ \lVert \nabla F_k(\mathbf{w})-\nabla F_k(\bar{\mathbf{w}}) \rvert \leq \beta \lvert \mathbf{w}-\bar{\mathbf{w}} \rvert$ for any $\mathbf{w}$, $\bar{\mathbf{w}}$ for any $\mathbf{w}$, $\bar{\mathbf{w}}$. \ignore{These assumptions will hold for ML models with convex loss function such as linear regression and SVM. By simulations, we will show that the proposed solutions work for non-convex models such as the neural networks with ReLU activation. }

It has been shown that for such a model, the difference between the global loss and the auxiliary loss at any iteration $l$ within an interval $g$, for $l = 1,\ldots,L$ and $g = 1,\ldots,G$, can be related to the local loss of a learner in the following way:
\begin{multline}
\left\lVert \mathbf{w}[l+1]-\hat{\mathbf{w}}[l+1] \right\rVert \leq \left\lVert \mathbf{w}[l]-\hat{\mathbf{w}}[l] \right\rVert + \\ \dfrac{\eta\beta}{d}\sum_{k=1}^{K} f_k\left[l-g\tau_m\right]
\label{staleness_div}
\end{multline}
Going back to our assumption that the global aggregation happens such that one or a set of learners have performed $\tau_m$ local updates since the previous global aggregation which is given by $t$, we can re-write the difference expression as:
\begin{multline}
\left\lVert \mathbf{w}[l+1]-\hat{\mathbf{w}}[l+1] \right\rVert \leq \left\lVert \mathbf{w}[l]-\hat{\mathbf{w}}[l] \right\rVert + \\ \dfrac{\eta\beta}{d}\sum_{k=1}^{K} f_k\left[t-\tau_k\right]
\label{staleness_div2}
\end{multline}
As it can be observed from the second term on the right-hand side in \ref{staleness_div}2, the model divergence is dependent upon the contributions from the local loss functions. 

It is expected that learners that have performed less local updates will have a higher loss and therefore, the model parameters will be further away from the optimal set. Hence, the scenario where we have many learners that have performed a low number of updates compared to the best performer, it is expected that the loss in general will be high. However, when all learners in general have performed a higher number of updates, the impact of staleness maybe lowered and the synchronous model of \cite{Mohammad2019a} may outperform the proposed architecture. These hypothesis are difficult to prove analytically but have been demonstrated experimentally.
\ignore{It is very difficult to prove these hypotheses theoretically, therefore we have shown this experimentally by analyzing both, the staleness for the proposed asynchronous model and the validation accuracy of three approaches, HA synchronous and asynchronous, and HU asynchronous.}

\ignore{Let us also assume that the local loss function at $F_k(\mathbf{w})$ does not diverge by more than $\delta_k$ such that $\lvert F_k(\mathbf{w})- F(\mathbf{w}) \rvert \leq \delta_k$ and $\delta = \frac{\sum_k d_k \delta_k}{d}$. Furthermore, $\lvert \mathbf{w} - \mathbf{v}_{[g]}[l-1] \rvert \leq h(l-(g-1)\tau)$. 
Then, it is shown in \cite{Tuor_01_AdaptiveControl}}

\section{Proof of Theorem \ref{theorem1}}
From the KKT optimality conditions, we have the following condition on the Lagrangian in (\ref{lagrangian}):
\begin{multline}
\nabla L_{z, \mathbf{\tau}, \mathbf{d}} = \nabla z + \\ \sum_{k = 1}^{K} \lambda_k \nabla \left(C_k^2  \tau_k d_k + C_k^1 d_k + C_k^0 - T\right) - \sum_{k=1}^K \nabla\alpha_k\tau_k + \\ \sum_{k=1}^K \nu_k \nabla  \left(-d_k+d_l\right) + \sum_{k=1}^K \nu_k^{\prime} \nabla \left(d_k-d_u\right) + \\ \sum_{n=1}^N \mu_n \nabla \left(-z+\tau_{c_{n,1}}-\tau_{c_{n,2}}\right) + \\ \sum_{n=1}^N \mu_n^{\prime} \nabla \left(-z-\tau_{c_{n,1}}+\tau_{c_{n,2}}\right) + \\ \omega\ \nabla \left(\sum_{k = 1}^{K}d_k - d\right)   = 0 \label{Eq_20_KKTlast}      
\end{multline}
The following sets of equations can be obtained after applying the derivatives for $\tau_k$ and $d_k$ in terms of the Lagrange multipliers, respectively, as shown in (\ref{eq:tk_lagrange}) and (\ref{eq:dk_lagrange}).
\begin{equation}
\lambda_k C_k^2 \tau_k^* + \lambda_k C_k^1 + \nu_k + \nu_k^\prime + \omega = 0, ~ \forall k
\label{eq:tk_lagrange}
\end{equation}
\begin{equation}
\lambda_k C_k^2 d_k^* + u_k + u_k^\prime + \alpha_k = 0 , ~ \forall k
\label{eq:dk_lagrange}
\end{equation}
Solving for $\tau_k^*$ and $d_k^*$ will give the results shown in (\ref{Eq_36_AnalBound}) and (\ref{Eq_30_BoundOnDk}). The procedure to obtain $u_k$ and $u_k^\prime$ is given in Appendix B.

\section{Obtaining $\mu$ and $\mu^\prime$}
The maximum staleness constraint in (\ref{relaxed-staleness-const}) can be re-written as two separate inequalities as shown below:
\begin{equation}
-z+\tau_k-\tau_l\leq 0
\label{Eq_ineq1}
\end{equation}
\begin{equation}
-z-\tau_k+\tau_l\leq 0
\label{Eq_ineq2}
\end{equation}
The $k^{th}$ element of the vector $\mathbf{u}$ denoted as $u_k$ is associated with the lagrange multipliers of the maximum staleness constraint inequality in (\ref{Eq_ineq1}) whereas $u_k^prime$ is associated with the inequality in (\ref{Eq_ineq2}), and the way to calculate them is shown in (\ref{Eq_del_uk}) and (\ref{Eq_del_uk_prime}), respectively.
\begin{equation}
u_k = \nabla_{\tau_k} \sum_{n=1}^N \mu_n\left(-z+\tau_k-\tau_l\right) 
\label{Eq_del_uk}
\end{equation}
\begin{equation}
u_k^\prime = \nabla_{\tau_k} \sum_{n=1}^N \mu_n^\prime\left(-z-\tau_k+\tau_l\right)
\label{Eq_del_uk_prime}
\end{equation}
As defined earlier, $k \in \kappa$ and $l \in \left\{\kappa \mid l > k ~ \forall~k \right\}$. 

In this case, after some manipulations, $u_k$ can be defined as the following:
\begin{equation}
u_k = \sum_{j=n_k}^{N_k} \mu_j - \sum_{j=1}^{K-1} \mu_{n_j+(k-j)}
\label{Eq_uksol}
\end{equation}
The start index and end indices of the first summation in (\ref{Eq_uksol}) are defined in (\ref{Eq_start_u}) and (\ref{Eq_end_u}), respectively.
\begin{equation}
n_k = 1+\sum_{m=0}^{k-1}\left(K-m\right)
\label{Eq_start_u}
\end{equation}
\begin{equation}
N_k = \sum_{m=1}^k \left(K-m\right)
\label{Eq_end_u}
\end{equation}

On the other hand, $u_k^\prime$ can be simply be defined as the following:
\begin{equation}
u_k^\prime = -\sum_{j=n_k}^{N_k} \mu_j^\prime + \sum_{j=1}^{K-1} \mu_{n_j+(k-j)}^\prime
\label{Eq_uksol_prime}
\end{equation}
\ignore{
\section*{Acknowledgment}
The authors would like to thank...}

\balance
\bibliographystyle{IEEEtran}
\raggedright
\bibliography{191RefsList}

\begin{thebibliography}{10}
\providecommand{\url}[1]{#1}
\csname url@samestyle\endcsname
\providecommand{\newblock}{\relax}
\providecommand{\bibinfo}[2]{#2}
\providecommand{\BIBentrySTDinterwordspacing}{\spaceskip=0pt\relax}
\providecommand{\BIBentryALTinterwordstretchfactor}{4}
\providecommand{\BIBentryALTinterwordspacing}{\spaceskip=\fontdimen2\font plus
\BIBentryALTinterwordstretchfactor\fontdimen3\font minus
  \fontdimen4\font\relax}
\providecommand{\BIBforeignlanguage}[2]{{%
\expandafter\ifx\csname l@#1\endcsname\relax
\typeout{** WARNING: IEEEtran.bst: No hyphenation pattern has been}%
\typeout{** loaded for the language `#1'. Using the pattern for}%
\typeout{** the default language instead.}%
\else
\language=\csname l@#1\endcsname
\fi
#2}}
\providecommand{\BIBdecl}{\relax}
\BIBdecl

\bibitem{N202003_vXchnge_IoT_data}
\BIBentryALTinterwordspacing
K.~Gyarmathy, ``{Comprehensive Guide to IoT Statistics You Need to Know in
  2020},'' 2020. [Online]. Available:
  \url{https://www.vxchnge.com/blog/iot-statistics}
\BIBentrySTDinterwordspacing

\bibitem{Ref2_Original}
\BIBentryALTinterwordspacing
{Rhea Kelly}, ``{Internet of Things Data To Top 1.6 Zettabytes by 2020 --
  Campus Technology},'' 2015. [Online]. Available:
  \url{https://campustechnology.com/articles/2015/04/15/internet-of-things-data-to-top-1-6-zettabytes-by-2020.aspx}
\BIBentrySTDinterwordspacing

\bibitem{P0_Letaief_Survey}
\BIBentryALTinterwordspacing
Y.~Mao, C.~You, J.~Zhang, K.~Huang, and K.~B. Letaief, ``{A Survey on Mobile
  Edge Computing: The Communication Perspective},'' \emph{IEEE Communications
  Surveys {\&} Tutorials}, vol.~19, no.~4, pp. 2322--2358, 2017. [Online].
  Available: \url{http://arxiv.org/abs/1701.01090}
\BIBentrySTDinterwordspacing

\bibitem{RefsPap_MEC126}
\BIBentryALTinterwordspacing
C.~You and K.~Huang, ``{Mobile Cooperative Computing: Energy-Efficient
  Peer-to-Peer Computation Offloading},'' pp. 1--33, 2017. [Online]. Available:
  \url{http://arxiv.org/abs/1704.04595}
\BIBentrySTDinterwordspacing

\bibitem{Moha1812:Multi}
U.~Y. Mohammad and S.~Sorour, ``{Multi-Objective Resource Optimization for
  Hierarchical Mobile Edge Computing},'' in \emph{2018 IEEE Global
  Communications Conference: Mobile and Wireless Networks (Globecom2018 MWN)},
  Abu Dhabi, United Arab Emirates, dec 2018.

\bibitem{Tuor_01_AdaptiveControl}
\BIBentryALTinterwordspacing
S.~Wang, T.~Tuor, T.~Salonidis, K.~K. Leung, C.~Makaya, T.~He, and K.~Chan,
  ``{When Edge Meets Learning : Adaptive Control for Resource-Constrained
  Distributed Machine Learning},'' in \emph{INFOCOM}, 2018. [Online].
  Available:
  \url{https://researcher.watson.ibm.com/researcher/files/us-wangshiq/SW{\_}INFOCOM2018.pdf}
\BIBentrySTDinterwordspacing

\bibitem{Wang2019}
\BIBentryALTinterwordspacing
------, ``{Adaptive Federated Learning in Resource Constrained Edge Computing
  Systems},'' \emph{IEEE Journal on Selected Areas in Communications}, no.
  Early Access, pp. 1--1, 2019. [Online]. Available:
  \url{https://ieeexplore.ieee.org/document/8664630/}
\BIBentrySTDinterwordspacing

\bibitem{2019arXiv190907972C}
\BIBentryALTinterwordspacing
M.~Chen, Z.~Yang, W.~Saad, C.~Yin, H.~V. Poor, and S.~Cui, ``{A Joint Learning
  and Communications Framework for Federated Learning over Wireless
  Networks},'' \emph{arXiv e-prints}, p. arXiv:1909.07972, sep 2019. [Online].
  Available: \url{https://arxiv.org/abs/1909.07972}
\BIBentrySTDinterwordspacing

\bibitem{2019arXiv191102417Y}
\BIBentryALTinterwordspacing
Z.~Yang, M.~Chen, W.~Saad, C.~S. Hong, and M.~Shikh-Bahaei, ``{Energy Efficient
  Federated Learning Over Wireless Communication Networks},'' \emph{arXiv
  e-prints}, p. arXiv:1911.02417, nov 2019. [Online]. Available:
  \url{https://arxiv.org/abs/1911.02417v1}
\BIBentrySTDinterwordspacing

\bibitem{Mohammad2019a}
\BIBentryALTinterwordspacing
U.~Mohammad and S.~Sorour, ``{Adaptive Task Allocation for Mobile Edge
  Learning},'' in \emph{2019 IEEE Wireless Communications and Networking
  Conference Workshop (WCNCW)}.\hskip 1em plus 0.5em minus 0.4em\relax IEEE,
  apr 2019, pp. 1--6. [Online]. Available:
  \url{https://ieeexplore.ieee.org/document/8902527/}
\BIBentrySTDinterwordspacing

\bibitem{Moha2020b}
U.~Y. Mohammad, S.~Sorour, and M.~S. Hefeida, ``Task allocation for mobile
  federated and offloaded learning with energy and delay constraints,'' in
  \emph{IEEE ICC 2020 Workshop on Edge Machine Learning for 5G Mobile Networks
  and Beyond (IEEE ICC'20 Workshop - EML5G)}, Dublin, Ireland, Jun. 2020.

\bibitem{StalenessAwarePaper}
Z.~Wei, S.~Gupta, X.~Lian, and J.~Liu, ``{Staleness-Aware Async-SGD for
  distributed deep learning},'' \emph{IJCAI International Joint Conference on
  Artificial Intelligence}, vol. 2016-Janua, pp. 2350--2356, 2016.

\bibitem{AsyncFedOpt_proof}
\BIBentryALTinterwordspacing
C.~Xie, S.~Koyejo, and I.~Gupta, ``Asynchronous federated optimization,''
  \emph{CoRR}, vol. abs/1903.03934, 2019. [Online]. Available:
  \url{http://arxiv.org/abs/1903.03934}
\BIBentrySTDinterwordspacing

\bibitem{QIP_NP_ArXiV}
A.~D. Pia, S.~S. Dey, and M.~Molinaro, ``{Mixed-integer Quadratic Programming
  is in NP},'' pp. 1--10, 2014.

\bibitem{MNIST_IEEE}
Y.~LeCun, L.~Bottou, Y.~Bengio, and P.~Haffner, ``{Gradient-based Learning
  Applied to Document Recognition},'' \emph{Proceedings of IEEE}, vol.~86,
  no.~11, 1998.

\end{thebibliography}

\ignore{
	\begin{IEEEbiography}{Yuguang ``Michael'' Fang}
Biography text here.
\end{IEEEbiography}
}


\end{document}